\documentclass{article}

\usepackage[utf8]{inputenc}
\usepackage{amsmath}
\usepackage{amsthm}
\usepackage{amssymb}
\usepackage{thmtools}
\usepackage{authblk}
\usepackage{algpseudocode}
\usepackage{algorithm}
\usepackage[dvipsnames]{xcolor}

\newtheorem{lemma}{Lemma}[section]

\newtheorem{theorem}{Theorem}[section]
\newtheorem{claim}{Claim}[section]

\theoremstyle{remark}
\newtheorem{case}{Case}[]

\title{Single Sample Prophet Inequality for Uniform Matroids of Rank $2$}
\author{Kanstantsin Pashkovich, Alice Sayutina}
\affil{University of Waterloo\\
Department of Combinatorics \& Optimization\\
}

\begin{document}
\date{}\maketitle
\begin{abstract}
    We study the prophet inequality when the gambler has an access only to a single sample from each distribution. Rubinstein, Wang and Weinberg showed that an optimal guarantee of $1/2$ can be achieved when the underlying matroid has rank $1$, i.e. in the case of a single choice. We show that this guarantee can be achieved also in the case of a uniform matroid of rank $2$ by a deterministic mechanism, and we show that this is best possible among deterministic mechanisms. We also conjecture that a straightforward generalization of our policy achieves the guarantee of $1/2$ for all uniform matroids.
\end{abstract}
\section{Introduction}

We study the single-sample prophet inequalities (SSPI) for uniform matroids. This is a variation of the prophet inequalities problem where the gambler does not know the distributions $X_1, X_2, \ldots, X_n$ for the arriving items, but has access only to a single sample $s_1 \sim X_1$, $s_2 \sim X_2$, \ldots, $s_n \sim X_n$ from each of the distributions.
After getting access to the samples, the gambler is presented realizations $r_1 \sim X_1$, $r_2 \sim X_2$, \ldots, $r_n \sim X_n$, and needs to decide whether to accept the elements in the online fashion. 
The problem is to design a mechanism  such that the expected value of the items accepted by the gambler is a good approximation of the value of the items accepted by the prophet, where the expectation is calculated with respect to the samples and 
realizations.

Let $k$ be the rank of the uniform matroid defining the feasibility constraints. We propose the following mechanism for the gambler. Here and later, we assume that all the values considered over all items have a perfect total order. In case when there are items with the same sample or realization values, we tiebreak them arbitrarily,
e.g. based on a randomly sampled $[0, 1]$ tie-breaker for each value.
We also assume that all the values are nonnegative. Indeed, neither the prophet nor the gambler accepts items with negative values, and so it is sufficient to design a mechanism only in the setting where no item has a negative value. 

\begin{algorithm}
\begin{algorithmic}[1]
\caption{Deterministic mechanism for uniform matroids of rank $k$.}\label{alg}
\State Obtain samples of the items' values $s_1 \sim X_1$, $s_2 \sim X_2$, \ldots, $s_n \sim X_n$.

\State Define $T$  as the $k$-th largest value among $s_1$, \ldots, $s_n$ if $n\geq k$. Define $T$ as~$0$ if $n< k$.

\State When item $i$, $i=1,\ldots,n$ arrives, accept it if $r_i>T$  and the total number of items accepted so far is less than $k$.

\end{algorithmic}
\end{algorithm}

 Note that in the case $k=1$ the mechanism in Algorithm~\ref{alg} behaves the same way as the $2$-competitive mechanism for uniform rank-1 matroids as studied in~\cite{rubinstein}. We prove that the mechanism in Algorithm~\ref{alg}
is  $2$-competitive also in the case $k=2$. We conjecture that this mechanism is $2$-competitive for $k > 2$ as well. 
\begin{theorem}\label{thm:main}
    The mechanism in Algorithm~\ref{alg} is $2$-competitive when $k=2$.
\end{theorem}
Moreover, the model used in the proof of Theorem~\ref{thm:main} shows that Algorithm~\ref{alg} is ``pointwise" $2$-competitive as  defined in~\cite{caramanisdutting} when $k=2$. ``Pointwise" SSPI inequalities allow samples and realizations for the same item to be correlated.

\subsection{Previous works}
Prophet inequalities were extensively studied in the context of full information about the values' distributions~\cite{chawla2010multi},\cite{dutting2015polymatroid}, \cite{gravin2019prophet}, \cite{kw2012matroid}, \cite{krengel1977semiamarts}, \cite{samuelcahn1984comparison}.

In~\cite{rubinstein}, it was shown that in the case of single-choice, i.e. in the case of uniform matroids of rank $1$, the optimal possible ratio of $2$ is achieved in the SSPI setting. Also, in~\cite{rubinstein} it was shown that when all distributions are identical then for any $\varepsilon>0$ the ratio $\approx 0.745$ can be achieved within the factor $(1+\varepsilon)$ with access to $O(n)$ samples, where the ratio  $\approx 0.745$ is also the optimal ratio for the full information case.

The results of~\cite{rubinstein} for uniform matroids of rank $1$ were crucial for~\cite{caramanis2021revisited} to improve  SSPI guarantees for several well studied downwards closed families.

For the choice of $k$ elements, i.e. in the case of uniform matroids of rank $k$, in~\cite{azar} it was provided a mechanims for SSPI with competitive guarantee $O(1-\frac{1}{\sqrt{k}})$. This is asymptotically matching the best known competitive guarantee in the case of prophet inequalities with full information for uniform matroids of rank $k$~\cite{azaruniform}.

In~\cite{azar}, the authors provided a blackbox reductions showing that if a family of downwards closed constraints admits an \emph{order oblivious secretary} mechanism with competitive guarantee $\alpha$ then the same family admits a ``pointwise" SSPI with the same competitive guarantee $\alpha$. Together with the fact that many known mechanisms for the secretary problem are order oblivious \cite{plaxton}, \cite{jaillet},  \cite{korula}, \cite{ma}, \cite{soto2013matroid} the blackbox reduction led to a series of SSPI mechanisms for different matrioids. Recently, in~\cite{caramanisdutting} it was shown that also the ``reverse" blackbox reduction holds, i.e. if a family of downward closed constraints  admits a ``pointwise" SSPI with competitive guarantee $\alpha$ then the same family admits an order oblivious secretary mechanism with competitive guarantee $2\alpha$.

\subsection{Model}
We make use of a standard model for SSPI, e.g. the same model was used in~\cite{rubinstein}. We assume  that for each item $i$, $i=1,\ldots,n$ we have two given values $y_i$ and $z_i$, where $y_i > z_i$.
Then with probability $1/2$ we have that $(s_i, r_i)$ equals  $(y_i, z_i)$, and with probability $1/2$ we have that $(s_i, r_i)$ equals
$(z_i, y_i)$. Every $\alpha$-competitive mechanism in such model is also $\alpha$-competitive in the general SSPI setting.

Let us consider all values $y_i$, $z_i$, $i=1,\ldots,n$, sorted in descending order.
Let this sequence be $w_1, w_2, \ldots, w_{2n}$, i.e. $w_1 > w_2 > \ldots > w_{2n}$.

Recall that in the considered model for every $i$, $i=1,\ldots,n$ exactly one of the values among $y_i$ and $z_i$ is picked to be equal to $s_i$ while the other to be equal to $r_i$. We refer to the indices in the sequence $w_1, w_2, \ldots, w_{2n}$ as \textit{elements}. Given an element $j$, $j=1,\ldots,2n$ if $w_j$ is picked to be equal to $s_i$ for some $i$, $i=1,\ldots,n$, we refer to $j$ as an \textit{\texttt{S}-element}. Similarly,  if $w_j$ is picked to be equal to  $r_i$ for some $i$, $i=1,\ldots,n$, we refer to $j$ as an  \textit{\texttt{R}-element}. Similarly, each element $j$, $j=1,\ldots,2n$ can be either a \textit{\texttt{Y}-element} or \textit{\texttt{Z}-element}
depending on whether $w_j$ is $y_i$ or $z_i$ for some $i$, $i=1,\ldots,n$.

We say that two values or two elements in the sequence $w_1, w_2, \ldots, w_{2n}$ are \textit{paired}  if they are $y_i$ and $z_i$ for the same item $i$, $i=1,\ldots,n$.

Let $j^*$ be the smallest \texttt{Z}-element, and $k^*$ be the second smallest \texttt{Z}-element among $1,\ldots, 2n$. Thus all elements $\{1, 2, \ldots, k^*-1\} \setminus \{j^*\}$ are \texttt{Y}-elements. 
Let $j^y$ be the \texttt{Y}-element, which is paired with $j^*$, and let $k^y$ be the \texttt{Y}-element, which is paired with $k^*$.
Naturally, all $j^*$, $j^y$, $k^*$, $k^y$ are distinct and satisfy
$j^y < j^*$, $k^y < k^*$.

For the sake of exposition, for $j=1,\ldots, 2n$ we say that a person, i.e. the gambler or the prophet, \textit{accepted element $j$} if this person accepted the item with the value $w_j$. 

\subsection{Bad Example}

Let us show that there is no deterministic mechanism for SSPI that achieves the competitiveness ratio better than $2$ on uniform matroids of rank $2$.

Let us assume that the gambler observes three items with following sample values $s_1=m$, $s_2=m$, $s_3=0$ where $m>0$. After that the gambler sees two realizations $r_1=m$ and $r_2=m$ before the realization of values for the third item is revealed.

Let us assume that in this scenario the gambler accepts $\beta$ items before the value of the third item is revealed.

If $\beta\leq 1$ then the mechanism cannot be $2$-competitive in the case, when the value of the first two items is $m$ with probability $1$ and the value of the last  item is $0$ with probability $1$. Indeed, the expected gain of the prophet in this case is $2m$, while the gain of the gambler is $\beta m$.

If $\beta=2$ then the mechanism cannot be  $\alpha$-competitive for $\alpha<2$, when the value of the first two items are $m$ with probability $1$ and the value of the third item is $0$ or $M$ with probabilities $1/2$ and $1/2$, respectively, where $M\gg m$. Indeed, the expected gain of the prophet in this case is at least $ M/2$,
while the gain of the gambler is at most 
\[
1/2 \cdot(2m)+1/2\cdot\left(1/2\cdot(2m)+ 1/2\cdot(m+M)\right)=1/4\cdot M+7/4\cdot m\,.
\]

\subsection{Expected Gain of Prophet}

We characterize the expected gain of the prophet by computing for each  element in the sequence $w_1$, $w_2$, \ldots, $w_{2n}$  the probability that the prophet accepts this element. 

\begin{lemma}\label{lm:prophet_prob} For each $j=1,\ldots, 2n$, the prophet accepts the element $j$ with probability $p_j$, where

$$p_j := \begin{cases}
j{/}2^j=2\cdot\frac{j}{2^{j+1}} &\qquad \text{if}\quad 1 \le j < j^*\\ 
(j-1){/}2^{j-1} &\qquad \text{if}\quad j = j^* \\
1{/}2^{j-2} &\qquad \text{if}\quad j^* < j < k^* \\
1{/}2^{j-3} &\qquad \text{if}\quad\qquad j = k^* \\
0 &\qquad \text{otherwise} \,.\\
\end{cases}$$

Thus, the expected gain of the prophet equals $\sum_{j=1}^{2n} p_j w_j$.
\end{lemma}

\begin{proof}
The prophet selects two largest \texttt{R}-elements among $w_1$, $w_2$, \ldots, $w_{2n}$.

Thus, for every $j$, $j=1,\ldots,2n$ the item with the value $w_j$ is selected if and only if
$j$ is an \texttt{R}-element, and there is at most one \texttt{R}-element among $1, 2, \ldots, j - 1$.

For elements $j$ such that $j < j^*$, the event of $j$ being an \texttt{R}-element and the event of there being at most one \texttt{R}-element among $1, 2, \ldots, j - 1$, are independent. The element $j$ is an \texttt{R}-element with probability $1/2$, and there is at most one element among $1, 2, \ldots, j - 1$ with probability $(j - 1)\frac{1}{2^{j-1}} + \frac{1}{2^{j-1}} = j{/}2^{j-1}$. Since those events
are independent, the probability of $j$ being accepted is
$j{/}2^j$.

Now consider the element $j$ such that $j = j^*$. The event of $j$ being an \texttt{R}-element happens with probability $1/2$, and implies that the
element $j^y$ is an \texttt{S}-element. The rest of the elements among $1, 2, \ldots, j - 1$ are independently either \texttt{R}-elements or \texttt{S}-elements. Thus the probability of having at most one \texttt{R}-element among $1, 2, \ldots, j-1$, conditioned on the event of $j$ being an \texttt{R}-element, is $(j-1){/}2^{j-2}$.
Thus the probability that $j$ is selected is $(j-1){/}2^{j-1}$.

Now consider an element $j$ such that $j^* < j < k^*$. There is exactly one \texttt{R}-element among $j^y, j^*$. Thus, for the element $j$ to be accepted the rest of the elements among $1, 2, \ldots, j - 1$ have to be \texttt{S}-elements. This happens with
probability $1{/}2^{j-3}$, and the element $j$ is an \texttt{R}-element with probability $1/2$. Thus the probability that $j$ is accepted is $1{/}2^{j-2}$.

Now consider the element $j$ such that $j = k^*$. The event of $j$ being an \texttt{R}-element happens with probability $1/2$, and implies that the
element $k^y$ is an \texttt{S}-element. There is also
exactly one \texttt{R}-element among $j^y, j^*$. The remaining elements among $1, 2, \ldots, j - 1$ are independently either \texttt{R}-elements or \texttt{S}-elements, and so to have at most one \texttt{R}-element among $1, 2, \ldots, j-1$ all of those remaining elements must be \texttt{S}-elements.
Thus the probability that $j$ is selected is $1{/}2^{j-3}$.

Finally, for all elements $j$ such that $j > k^*$, there is exactly one \texttt{R}-element among $j^y, j^*$, and exactly one \texttt{R}-element among $k^y, k^*$.
Thus there are two \texttt{R}-element among $1, 2, \ldots, j - 1$, and so $j$ is never selected by the prophet.
\end{proof}

\subsection{Expected Gain of Gambler}

Recall that the adversary cannot control which elements are \texttt{R}-elements,
and which are \texttt{S}-elements. However, the adversary can control in which order
\texttt{R}-elements are presented to the gambler. We  assume that the adversary
behaves in a way which minimizes the gambler's gain.

Recall that in Algorithm~\ref{alg}, the gambler determines the threshold $T$
and accepts elements which are greater than $T$. Thus, the adversary makes the gambler to accept two smallest \texttt{R}-elements that are larger than $T$, if those elements exist. In case there are less than two such elements, the adversary  makes the gambler accept all \texttt{R}-elements that are larger than $T$ and no other elements. From now on, we assume that the adversary acts as above, i.e. in a way leading to the worst gain for the gambler.

\begin{lemma}\label{lm:gambler_prob} For each $j=1,\ldots, 2n$, the gambler accepts the element $j$ with probability at least $q_j$, where
$$q_j := \begin{cases}
(3j - 1){/}2^{j+2} & \qquad \text{if}\quad j \le j^* - 2 \\
(4j - 2){/}2^{j+2} &\qquad \text{if}\quad  j = j^* - 1\text{ and }k^*>j^*+1\text{ and } j^*-1=j^y\\
(4j - 3)/2^{j+2} &\qquad \text{if}\quad  j = j^* - 1\text{ and }k^*>j^*+1\text{ and } j^*-1\neq j^y\\
4j{/}2^{j+2} & \qquad \text{if}\quad j = j^* - 1\text{ and }k^*=j^*+1\\
3{/}2^{j+1} & \qquad \text{if}\quad j = j^*\text{ and }k^*>j^*+1  \\
4{/}2^{j+1} & \qquad \text{if}\quad j = j^*\text{ and }k^*=j^*+1  \\
3{/}2^{j} & \qquad \text{if}\quad j^* < j < k^*-1
\\
1{/}2^{j-2} & \qquad \text{if}\quad j^* < j < k^* \text{ and } j=k^*-1\,.\\
\end{cases}$$
Thus, the expected gain of the gambler is at least $\sum_{j=1}^{2n} q_j w_j$.
\end{lemma}

The key take-away from Lemma~\ref{lm:gambler_prob}
is that element $j$, $j < j^*$ is accepted with probability at least $(3j-1){/}2^{j+2}$. Element $j=j^*$ is accepted with probability at least $3{/}2^{j+1}$, and element $j$, $j^*<j<k^*$ is accepted with probability at least $3{/}2^{j}$. However, in some cases for the comparison with the gain of the prophet later we need a stronger bound, for example, in cases when $k^*=j^*+1$ or if $j = j^*-1$.

\begin{proof} [Proof of Lemma~\ref{lm:gambler_prob}] Let us consider the case analysis based on the position of the element $j$. In each of the cases, we compute the unconditional probabilities for two events to happen simultaneously, in particular for the gambler to accept the element $j$ and for the threshold to be equal to $T$. Afterwards, we provide a desired lower bound $q_j$ by summing up the obtained unconditional probabilities.

\begin{case} $j \le j^* - 4$.

Let us  consider the position of the threshold, i.e. the position of the second largest \texttt{S}-element.
\begin{enumerate}
\item $T=w_{j+3}$. For the threshold $T$ to be equal $w_{j+3}$, the element $j+3$ is an \texttt{S}-element and  exactly one element in $1$,\ldots, $j+2$ is an \texttt{S}-element, i.e. is the largest \texttt{S}-element. Let us now consider the position of the unique \texttt{S}-element in $1$,\ldots, $j+2$. If this unique \texttt{S}-element is not $j+1$ or  $j+2$, the gambler will not be able to accept the element $j$ due to the assumption that the adversary minimizes the gambler's profit. Thus, all elements in  $1$,\ldots, $j$ have to be \texttt{R}-elements, 
the element $j+3$ has to be an \texttt{S}-element, and among $j+1$ and $j+2$ exactly one is an \texttt{R}-element and one is an \texttt{S}-element. Also note, that any such \texttt{S}-\texttt{R} status assignment guarantees that $T$ equals $w_{j+3}$ and the gambler accepts the element $j$.

Let us now compute the probability of the event that $T$ equals $w_{j+3}$ and the gambler accepts the element $j$. There are two choices for the largest \texttt{S}-element. In each of those scenarios we fix the \texttt{S}-\texttt{R} status of $j+3$ elements. Thus the probability of the gambler accepting element $j$ and for the threshold $T$ to be equal $w_{j+3}$ is $2 /2^{j+3}$.

\begin{figure}
    \centering
    \footnotesize

    \colorbox{Apricot}{%
    \begin{tabular}{ccccccccc}

              &     &      & {\scriptsize \color{WildStrawberry} (*)} &       & {\color{WildStrawberry} \scriptsize (*)}  &  {\scriptsize \textcolor{red}{[T]} }&  &  \\
           R & R    & R    & R   & \textbf{S}     & R      &   \textbf{S} & R & \textbf{S} \\
       $j-3$ &$j-2$ &$j-1$ & $j$ & $j+1$ & $j+2$  & $j+3$ & $j+4$ & $j+5$  \\
    \end{tabular}%
    }
    \caption{An example of an \texttt{S}-\texttt{R} status assignment in subcase 1 of case 1, which makes the gambler accept item $j$. The threshold is denoted by \textcolor{red}{[T]}. Elements accepted by the gambler are denoted by \textcolor{WildStrawberry}{(*)}. The larger elements are at the left.}
\end{figure}

\item $T=w_{j+2}$. The largest \texttt{S}-element is either $j+1$ or is among  $1$, $2$, \ldots, $j-2$, $j-1$. Thus there are $j$ choices for the largest \texttt{S}-element, in each of these choices element $j$ will be accepted, and in each of these choices we fix the \texttt{S}-\texttt{R} status of $j+2$ elements. Thus the probability of the gambler accepting element $j$ and for the threshold $T$ to be equal $w_{j+2}$ is $j / 2^{j+2}$.

\item $T=w_{j+1}$. The largest \texttt{S}-element is among  $1$, $2$, \ldots, $j-2$, $j-1$. Thus there are $j-1$ choices for the largest \texttt{S}-element, in each of these choices  we fix the \texttt{S}-\texttt{R} status of $j+1$ elements. Thus the probability of the gambler accepting element $j$ and for the threshold $T$ to be equal $w_{j+1}$ is $(j-1) / 2^{j+1}$.
\end{enumerate}

Hence, the probability that the gambler accepts the element $j$ equals $$\frac{1}{2^{j+3}} (2 + 2j + 4j-4) = \frac{3j - 1}{2^{j+2}}\,.$$

\end{case}

\begin{case} $j = j^* - 3$

Let us again  consider the position of the threshold, i.e. the position of the second largest \texttt{S}-element.
\begin{enumerate}
\item $T = w_{j+1}$. Then the gambler  accepts element $j$ if the largest \texttt{S}-element is one of $1, 2, \ldots, j - 1$. Thus there are $j - 1$ choices for the largest \texttt{S}-element, and in each of those choices we fix the \texttt{S}-\texttt{R} status of $j + 1$ elements. Thus the probability of the gambler accepting element $j$ and $T$ being equal to $w_{j+1}$ is $(j-1){/}2^{j+1}$.

\item $T = w_{j+2}$. Then the gambler accepts element $j$ if the largest \texttt{S}-element is one of $1, 2, \ldots, j - 2, j - 1$ or $j + 1$. 
Thus there are $j$ choices for the largest \texttt{S}-element, in each of these choices  we fix the \texttt{S}-\texttt{R} status of $j+2$ elements. Thus the probability of the gambler accepting element $j$ and for the threshold $T$ to be equal $w_{j+2}$ is $j{/}2^{j+2}$.

\item $T = w_{j+3}$.
Then the gambler accepts element $j$ if the largest \texttt{S}-element is $j+1$ or $j+2$.
If we fix the element $j^* = j+3$ to be an \texttt{S}-element,  then the element $j^y$ is an \texttt{R}-element.
Thus if $j^y$ is in $\{j+1, j+2\}$ there is one choice for the largest \texttt{S}-element, and if $j^y$ is not in $\{j+1, j+2\}$ there are two choices.
In each of these choices we fix the \texttt{S}-\texttt{R} status of
the first $j + 3$ elements except for the element $j^y$, since its
\texttt{S}-\texttt{R} status is determined by the \texttt{S}-\texttt{R} status
of element $j^* = j+3$.
There is at least one choice for the largest \texttt{S}-element, and each such choice happens with probability $1{/}2^{j+2}$. Thus the probability for the gambler to accept element $j$ and for the threshold $T$ to be equal $w_{j+3}$ is at least $1{/}2^{j+2}$.
\end{enumerate}

Hence, the probability that the gambler accepts the element $j$ is at least $$\frac{1}{2^{j+2}} ((2j - 2) + j + 1) = \frac{3j - 1}{2^{j+2}}\,.$$
\end{case}

\begin{case}\label{case:3a} $j = j^* - 2$, $k^* > j^* + 1$

\begin{enumerate}
\item $T = w_{j + 1}$. Then the gambler accepts element $j$ if the largest \texttt{S}-element is one of $1$, $2$, \ldots, $j - 1$. Thus there are $j - 1$ choices for the largest \texttt{S}-element, and in each of those choices we fix the \texttt{S}-\texttt{R} status of $j + 1$ elements. Thus the probability of the gambler accepting element $j$ and $T$ being equal to $w_{j+1}$ is $(j-1){/}2^{j+1}$.

\item \label{case1} $T = w_{j + 2}$. Then the gambler accepts element $j$ if the largest \texttt{S}-element is one of $1$, $2$, \ldots, $j - 2$, $j - 1$ or $j + 1$.
Since $j^* = j+2$ is an \texttt{S}-element, the element $j^y$ is an \texttt{R}-element.

If $j^y \in \{1, 2, \ldots, j - 2, j - 1\}\cup\{j + 1\}$ there are $j - 1$ choices for the largest \texttt{S}-element. In this case, the probability that the gambler accepts element $j$ and $T$ is equal to $w_{j+2}$ is  $(j-1){/}2^{j+1}$.

However if  $j^y \not\in \{1, 2, \ldots, j - 2, j - 1\}\cup\{j + 1\}$, there are $j$ choices and so  the probability that the gambler accepts element $j$ and $T$ is equal to $w_{j+2}$ is $j{/}2^{j+1}$.

 Thus, the probability that the gambler accepts element $j$ and $T$ is equal to $w_{j+2}$ is at least $(j-1){/}2^{j+1}$.

\item \label{case2} $T = w_{j + 3}$. The gambler accepts element $j$ only if the largest \texttt{S}-element is $j + 1$ or $j + 2$.
Recall that exactly one of $j^* = j+2$ and $j^y$ is an \texttt{S}-element. Thus, if $j^y = j + 1$ then we have two  choices for the largest \texttt{S}-element. If $j^y < j + 1$ then we have only one  choice for the largest \texttt{S}-element.

Thus there is one or two choices for the largest \texttt{S}-element, and each time we fix the \texttt{S}-\texttt{R} status of $j + 3$ elements where two of these elements are paired. Thus, the probability, that the gambler accepts element $j$ and that $T$ is equal to $w_{j+3}$, is $1{/}2^{j+2}$ or $2{/}2^{j+2}$, and so is at least $1{/}2^{j+2}$.
\end{enumerate}

Hence, the probability that the gambler accepts the element $j$ is at least $$\frac{1}{2^{j+2}} \left((2j - 2) + (2j - 2) + 1\right) = \frac{4j - 3}{2^{j+2}}\,.$$

Note, that 
$4j - 3\ge 3j - 1$ when $j \ge 2$. So when $j\geq 2$ we get the desired bound. 

When $j = 1$, observe that it is not possible that the estimate $(j-1){/}2^{j+1}$  is tight in second subcase and  the estimate $1{/}2^{j+2}$ is tight in third subcase, because for this in second subcase we need $j^y = j + 1 = 2$ and in third subcase $j^y \ne j + 1 = 2$.
Thus in case $j = 1$, we can derive a tighter bound. The probability that the gambler accepts  element $j$ is at least
$$\frac{1}{2^{j+2}} \min((2j - 2) + 2j + 1, (2j - 2) + (2j - 2) + 2) = \frac{4j - 2}{2^{j+2}} \ge \frac{3j - 1}{2^{j+2}}\,,$$
showing that the desired lower bound holds also for $j=1$.
\end{case}

\begin{case} $j = j^* - 2$, $k^* = j^* + 1$
\begin{enumerate}
\item $T = w_{j + 1}$. This subcase is identical to the first subcase in case~\ref{case:3a}. The probability of the gambler accepting element $j$ and $T$ being equal to $w_{j+1}$ is $(j-1){/}2^{j+1}$.

\item $T = w_{j + 2}$. This subcase is identical to the second subcase in case~\ref{case:3a}. The probability of the gambler accepting element $j$ and $T$ being equal to $w_{j+2}$ is at least $(j-1){/}2^{j+1}$.

\item $T = w_{j + 3}$. Then the gambler  accepts element $j$ if the largest \texttt{S}-element is $j + 1$ or $j + 2$.
However, observe exactly one of $j^* = j+2$ and $j^y$, as well as exactly one of $k^*=j+3$ and $k^y$, is an \texttt{S}-element.
Thus, if $j^y = j + 1$ then we have two  choices for the largest \texttt{S}-element. If $j^y < j + 1$ then we have only one  choice for the largest \texttt{S}-element.

Thus there is at least one choice, and each choice happens with  probability
$1{/}2^{j+1}$, since we are fixing the \texttt{S}-\texttt{R} status of the first $j + 3$ elements, among which there are paired elements $j^*$, $j^y$ and $k^*$, $k^y$.
Then the probability of the gambler accepting element $j$ and $T$ being equal to $w_{j+3}$ is at least $1{/}2^{j+1}$.
\end{enumerate}

Hence, the probability that the gambler accepts the element $j$ is at least 
$$\frac{1}{2^{j+1}} ((j - 1) + (j - 1) + 1) = \frac{4j - 2}{2^{j+2}} \ge \frac{3j - 1}{2^{j+2}},$$
where the last inequality follows from $j\geq 1$.
\end{case}


\begin{case}\label{case:4ab} $j = j^* - 1$, $k^* > j^* + 1$, $j^*-1=j^y$

\begin{enumerate}
\item $T = w_{j+1}$. Then the gambler accepts element $j$ if the largest \texttt{S}-element is one of $1$, $2$, \ldots, $j - 1$. Thus there are $j - 1$ choices for the largest \texttt{S}-element, and in each of those choices we fix the \texttt{S}-\texttt{R} status of $1$, $2$, \ldots, $j$, $j+1$. Note that $j$ and $j+1$ are paired. Thus the probability of the gambler accepting element $j$ and $T$ being equal to $w_{j+1}$ is $(j-1){/}2^{j}$.

\item $T = w_{j+2}$. For this subcase, $j + 1$ needs to be the largest \texttt{S}-element. By making $j + 1$ the largest \texttt{S}-element,  we fix the \texttt{S}-\texttt{R} status of $1$, $2$, \ldots, $j$, $j+1$, $j+2$. Note that $j$ and $j+1$ are paired. Thus the probability of the gambler accepting element $j$ and $T$ being equal to $w_{j+2}$ is $1{/}2^{j+1}$.

\item $T = w_{j+3}$. Although it is possible for the threshold to be at the position $j+3$ when  item $j$ is accepted, we omit this subcase since the total probability from the previous subcases is already sufficient. We remark that the current subcase $T = w_{j+3}$ has different probability depending on whether $k^* = j^* + 2$ or $k^* \ge j^*+3$.
\end{enumerate}

Hence, the probability that the gambler accepts the element $j$ is at least 
\[\frac{1}{2^{j+1}} \left(2(j - 1) + 1\right) = \frac{4j - 2}{2^{j+2}}\,.\]
\end{case}

\begin{case}\label{case:4cd} $j = j^* - 1$, $k^* > j^* + 1$, $j^*-1\neq j^y$

\begin{enumerate}
\item $T = w_{j+1}$. Then the gambler accepts element $j$ if the largest \texttt{S}-element is one of $1$, $2$, \ldots, $j - 1$ but not $j^y$. Thus there are $j - 2$ choices for the largest \texttt{S}-element, and in each of those choices we fix the \texttt{S}-\texttt{R} status of $1$, $2$, \ldots, $j$, $j+1$. Note that $j+1$ and $j^y$ are paired. Thus the probability of the gambler accepting element $j$ and $T$ being equal to $w_{j+1}$ is $(j-2){/}2^{j}$.

\item $T = w_{j+2}$. Note that $j+1$ and $j^y$ are paired, so for this subcase element $j+1$ or $j^y$ need to be the largest \texttt{S}-element. We fix the \texttt{S}-\texttt{R} status of $1$, $2$, \ldots, $j$, $j+1$, $j+2$, where all elements that are not $j+2$, $j+1$ or $j^y$ become \texttt{R}-elements and where the \texttt{S}-\texttt{R} statuses of  $j+1$, $j^y$ are arbitrary. Thus the probability of the gambler accepting element $j$ and $T$ being equal to $w_{j+2}$ is $1{/}2^{j}$.

\item $T = w_{j+3}$. Note that $j+1$ and $j^y$ are paired. So for this subcase element $j+1$ needs to be the largest \texttt{S}-element, otherwise $j$ is not accepted by the gambler. We fix the \texttt{S}-\texttt{R} status of $1$, $2$, \ldots, $j$, $j+1$, $j+2$, $j+3$ where  elements  $j+1$ and $j^y$ are paired. We remark that the current subcase $T = w_{j+3}$ has different probability depending on whether $k^* = j^* + 2$ or $k^* > j^*+2$. Independent on whether $k^*=j+2$, the probability of the gambler accepting element $j$ and $T$ being equal to $w_{j+3}$ is at least $1{/}2^{j+2}$. 
\end{enumerate}

 Thus, the probability that the gambler accepts element $j$ is at least 
\[\frac{1}{2^{j+1}} \left(2(j - 2) + 2+1/2\right) = \frac{4j -3}{2^{j+2}}\,.\]
\end{case}

\begin{case} $j = j^* - 1$, $k^* = j^* + 1$,  $j^*-1=j^y$
\begin{enumerate}
\item $T = w_{j+1}$. This case is identical to the first subcase of case~\ref{case:4ab}. The probability of the gambler accepting element $j$ and $T$ being equal to $w_{j+1}$ is $(j-1){/}2^{j}$.
\item $T = w_{j+2}$. This case is similar to the second subcase of case~\ref{case:4ab}, but where among $1$, $2$, \ldots, $j$, $j+1$, $j+2$ we have paired $j^*$, $j^y$ and paired $k^*$, $k^y$. Thus the probability of the gambler accepting element $j$ and $T$ being equal to $w_{j+2}$ is $1{/}2^{j}$.
\end{enumerate}

Hence, the probability that the gambler accepts  element $j$ is at least 
\[\frac{1}{2^{j+1}} (2(j - 1) + 2) = \frac{4j}{2^{j+2}}\,.\]
\end{case}

\begin{case} $j = j^* - 1$, $k^* = j^* + 1$,  $j^*-1\neq j^y$
This case is similar to case~\ref{case:4cd}, except that in the first subcase we get the estimate $(j-2)/2^j$, in the second $1/2^{j-1}$ and in the third case $0$. Thus, the probability that the gambler accepts the element $j$ is at least 
\[\frac{1}{2^{j+1}} \left(2(j - 2) + 4\right) = \frac{4j}{2^{j+2}}\,.\]
\end{case}

Before we move further, observe that in the case $j = j^*$, element $j$ can only be accepted if it is
an \texttt{R}-element. For this element $j^y$ needs to be an \texttt{S}-element.
Secondly, an element $j$ can be accepted only if it exceeds the threshold, which is the second largest \texttt{S}-element Thus
$j^y$ is the largest \texttt{S}-element, so unlike in the cases with $j < j^*$, there are no alternative choices for the largest \texttt{S}-element.

\begin{case}\label{case:5a} $j = j^*$, $k^* \ge j^* + 3$
\begin{enumerate}
\item $T = w_{j+1}$. This fixes the \texttt{S}-\texttt{R} status of $j+1$ elements $1$, $2$, \ldots, $j - 1$, $j$, $j + 1$. Note that among these elements the elements $j = j^*$ and $j^y$ are paired. Thus the probability of the gambler accepting  element $j$ and $T$ being equal to $w_{j+1}$ is $1{/}2^{j}$

\item $T = w_{j+2}$. This fixes the \texttt{S}-\texttt{R} status of $j+2$ elements $1$, $2$, \ldots, $j - 1$, $j$, $j + 1$, $j + 2$. Note that among these elements the elements $j = j^*$ and $j^y$ are paired. Thus the probability of the gambler accepting  element $j$ and $T$ being equal to $w_{j+2}$ is $1{/}2^{j+1}$
\end{enumerate}

Hence, probability that the gambler accepts the element $j$ is $3{/}2^{j+1}$.
\end{case}

\begin{case} $j = j^*$, $k^* = j^* + 2$
\begin{enumerate}
\item $T = w_{j+1}$. This subcase is identical to first subcase of case~\ref{case:5a}, so the probability is $1{/}2^{j}$ in this subcase.
\item $T = w_{j+2}$. This fixes the \texttt{S}-\texttt{R} status of $1$, $2$, \ldots, $j + 1$, $j+2$.  Note that among these elements the elements $j = j^*$ and $j^y$ are paired, also elements $j+2=k^*$ and $k^y$ are paired. Thus, the probability of the gambler accepting  element $j$ and $T$ being equal to $w_{j+2}$ is $1{/}2^{j}$.
\end{enumerate}

Hence, the probability that the gambler accepts element $j$ is $4{/}2^{j+1}$, which is at least the desired bound  $3{/}2^{j+1}$.
\end{case}

\begin{case} $j = j^*$, $k^* = j^* + 1$
\begin{enumerate}
\item $T = w_{j+1}$. This fixes the \texttt{S}-\texttt{R} status of elements $1$, $2$, \ldots, $j - 1$, $j$, $j+1$. Note that among these elements the elements $j = j^*$ and $j^y$ are paired, also elements $j+1=k^*$ and $k^y$ are paired.  Thus, the probability of the gambler accepting  element $j$ and $T$ being equal to $w_{j+1}$ is $1{/}2^{j-1}$.
\end{enumerate}

Hence, the probability that the gambler accepts element $j$ is $4{/}2^{j+1}$.
\end{case}

Now we consider the cases when $j^* < j < k^*$. For each of these cases, we need to show the desired lower bound $3{/}2^{j}$ for the probability of accepting item $j$.
As previously, we proceed with analysis based on the position of the threshold $T$, i.e. the position of the second largest \texttt{S}-element. Observe that there is exactly one \texttt{S}-element among $\{j^y, j^*\}$. Thus the threshold $T$ is less than $w_j$ only if there are no other \texttt{S}-elements
before $j$ except for the elements $j^y$ and  $j^*$.

\begin{case}\label{case:stage2_1} $j^* < j < k^*$, $k^* \ge j + 3$
\begin{enumerate}
\item $T = w_{j+1}$. For  element $j$ to be accepted we need to fix the \texttt{S}-\texttt{R} status of $j - 1$ elements $\{1,\ldots, j+1\}\setminus \{j^y, j^*\}$. Thus the probability of the gambler accepting element $j$ and $T$ being equal to $w_{j+1}$ is $1{/}2^{j-1}$.
\item $T = w_{j+2}$. For element $j$ to be accepted we need to fix the \texttt{S}-\texttt{R} status of $j$ elements $\{1,\ldots, j+2\}\setminus \{j^y, j^*\}$. Thus the probability of the gambler accepting element $j$ and $T$ being equal to $w_{j+1}$ is $1{/}2^j$.
\end{enumerate}

Hence, the probability that the gambler accepts element $j$ is $3{/}2^{j}$.
\end{case}

\begin{case} $j^* < j < k^*$, $k^* = j + 2$

\begin{enumerate}
\item $T = w_{j+1}$. This case is identical to the first subcase of case~\ref{case:stage2_1}. The
probability of the gambler accepting element $j$ and $T$ being equal to $w_{j + 1}$ is $1{/}2^{j-1}$
\item $T = w_{j+2}$. For the element $j$ to be accepted we need to fix the \texttt{S}-\texttt{R} status of $j-1$ elements $\{1,\ldots, j+2\}\setminus \{j^y, j^*, k^y\}$. The \texttt{S}-\texttt{R} status of element $k^y$ does not need to be fixed because it is paired with $j+2=k^*$. Thus the probability of the gambler accepting element $j$ and $T$ being equal to $w_{j+2}$ is $1{/}2^{j-1}$.
\end{enumerate}

Hence, the probability that the gambler accepts  element $j$ is $4{/}2^{j}$.

\end{case}

\begin{case} $j^* < j < k^*$, $k^* = j + 1$

\begin{enumerate}
\item $T = w_{j+1}$. Again for the element $j$ to be accepted we need to fix the \texttt{S}-\texttt{R} status of $j - 2$ elements $\{1,\ldots, j+1\}\setminus \{j^y, j^*, k^y\}$. Thus the probability, that the gambler accepts element $j$ and $T$ is equal to $w_{j+1}$, is $1{/}2^{j-2}$.
\end{enumerate}

Hence, the probability that the gambler accepts  element $j$ is $4{/}2^{j}$.

\end{case}
\end{proof}
\subsection{Putting Everything Together}

In this section, we provide the proof of Theorem~\ref{thm:main}.
 \begin{proof}[Proof of Theorem~\ref{thm:main}] 
 
 Let us show that \[2\sum_{j=1}^{2n} q_j w_j \ge \sum_{j=1}^{2n} p_j w_j\,.\]
Since $w_1$, $w_2$, \ldots, $w_{2n}$ is a non-increasing sequence of non-negative numbers, it is sufficient to show that for all $i$, $i=1,\ldots,n$  we have
\begin{equation}\label{eq:main}
    2\sum_{j=1}^i q_j \ge \sum_{j=1}^i p_j\,.
\end{equation}
Observe that it is sufficient to prove~\eqref{eq:main} only for $i \le k^*$, since $p_j = 0$ when~$j > k^*$.

    \begin{claim}\label{claim1} If $j<j^*$ or $j^* < j < k^*$, then we have $2 q_j \ge p_j$.
    \end{claim}
    Consequently, this claim directly proves~\eqref{eq:main} for $i \le j^* - 1$.
    \begin{proof}[Proof of Claim~\ref{claim1}]
    We use Lemma~\ref{lm:prophet_prob} and Lemma~\ref{lm:gambler_prob}.

    If $j < j^*$, then we have $p_j = j{/}2^j$. If $j < j^*$, then we have $q_j \ge (3j - 1){/}2^{j+2}$, or $j\geq 2$ and  $q_j = (4j - 3){/}2^{j+2}$; so in either case $q_j \ge (3j - 1){/}2^{j+2}$. Thus, we have
  \[q_j \ge (3j - 1){/}2^{j+2} \ge 2j{/}2^{j+2} = p_j/2\,,\]
    where the second inequality uses $j \ge 1$.

    If $j^* < j < k^*$, then we have $p_j = 4{/}2^j$ and  $q_j \geq 3{/}2^j$, and so  $q_j \ge  p_j/2$.
    \end{proof}

    \begin{claim}\label{claim2} If $ k^*>j^*+1$, then we have $2 q_{k^*-1} \ge p_{k^*}+p_{k^*-1}$.
    \end{claim}
\begin{proof}[Proof of Claim~\ref{claim2}]
We use Lemma~\ref{lm:prophet_prob} and Lemma~\ref{lm:gambler_prob} to show
    \[p_{k^*}+p_{k^*-1}=\frac{1}{2^{k^*-3}}+\frac{1}{2^{k^*-3}}=2\cdot\frac{1}{2^{k^*-3}}=2q_{k^*-1}\,.\]
\end{proof}

    \begin{claim}\label{claim3} We have $2(q_{j^* - 1} + q_{j^*}) \geq p_{j^* - 1} + p_{j^*}$. Moreover, if $k^*=j^*+1$ then we have $2(q_{j^* - 1} + q_{j^*}) \geq p_{j^* - 1} + p_{j^*}+p_{k^*}$.
    \end{claim}

    \begin{proof}[Proof of Claim~\ref{claim3}]
    We use Lemma~\ref{lm:prophet_prob} and Lemma~\ref{lm:gambler_prob}.

    If $k^* > j^* + 1$, then we have
    \begin{align*}
        p_{j^* - 1} + p_{j^*} =\frac{2j^*-2}{2^{j^*-1}}=
        2\left(\frac{4j^* - 7}{2^{j^*+1}} + \frac{3}{2^{j^*+1}} \right)\leq 2(q_{j^*-1}+q_{j^*})\,.
    \end{align*}

    If $k^* = j^* + 1$, we have
        \begin{align*}
        p_{j^* - 1} + p_{j^*}+p_{k^*}= \frac{2j^*-2}{2^{j^*-1}} + \frac{1}{2^{j^*-2}} =
        2\left(\frac{4j^*-4}{2^{j^*+1}}+ \frac{4}{2^{j^*+1}} \right)=2(q_{j^*-1}+q_{j^*})\,,
    \end{align*}
finishing the proof.
    \end{proof}

\end{proof}

\bibliography{literature}
\bibliographystyle{abbrv}
\end{document}